\documentclass[12pt,letterpaper]{article}
\usepackage{graphicx}
\usepackage{mathrsfs}
\usepackage{dsfont}
\usepackage{units}
\usepackage{amsmath, amsthm, amssymb, amsfonts, enumerate}
\usepackage[round,authoryear]{natbib}
\usepackage{color}
\usepackage{caption}
\usepackage{subcaption}
\usepackage{setspace}
\usepackage{mdframed}
\usepackage{hyperref}
\usepackage{float}
\usepackage{multirow}
\usepackage[normalem]{ulem}
\usepackage{placeins}
\usepackage{tikz}
\usepackage[T1]{fontenc}
\usepackage[utf8]{inputenc}
\usepackage{lmodern}
\usepackage[letterpaper,margin=1.25in]{geometry}
\usepackage{microtype}
\usepackage{needspace}
\hypersetup{hidelinks,pdftitle={Comparison of Deterministic Information Providers},
 pdfauthor={David Lagziel, Ehud Lehrer, Tao Wang}}
\usetikzlibrary{calc}

\definecolor{blue}{rgb}{0.0, 0.0, 1}
\definecolor{green2}{rgb}{0.0, 0.52, 0.24}
\definecolor{cadmiumgreen}{rgb}{0.0, 0.42, 0.24}
\definecolor{camouflagegreen}{rgb}{0.47, 0.53, 0.42}
\definecolor{darkolivegreen}{rgb}{0.33, 0.42, 0.18}
\definecolor{darkpastelgreen}{rgb}{0.01, 0.75, 0.24}
\definecolor{darkspringgreen}{rgb}{0.09, 0.45, 0.27}

\usepackage{comment}

\def\o{{\omega}}

\def\1{\textbf{1}}

\newtheorem{theorem}{Theorem}
\newtheorem{proposition}{Proposition}
\newtheorem{corollary}{Corollary}
\newtheorem{example}{Example}
\newtheorem{definition}{Definition}

\newtheorem{obs}{Observation}

\begin{document}

\title{Comparison of Deterministic Information Providers\thanks{For their valuable comments, the authors wish to thank participants of the Durham University Economics Seminar, the Adam Smith Business School Micro theory seminar of Glasgow University, INSEAD EPS seminar, The School of Economics seminar at the University of Edinburgh, Western University Economics seminar, the Tel-Aviv University Game Theory Seminar, the Rationality Center Game Theory Seminar, the Technion Game Theory Seminar, the Bar-Ilan University Theoretical Economics Seminar, the Bar-Ilan University Management Seminar and the BGU Economics seminar.
Lagziel acknowledges the support of the Israel Science Foundation, Grant \#2074/23. Lehrer acknowledges the support of the Deutsche Forschungsgemeinschaft (DFG, German Research Foundation), Project Number 461570745. Wang acknowledges the support of National Natural Science Foundation of China \#72303161.}}
\author{David Lagziel\thanks{\raggedright Department of Economics, Ben-Gurion University of the Negev, Beer-Sheba 8410501, Israel.  E-mail: \textsf{Davidlag@bgu.ac.il}.} \\
{\small Ben-Gurion University}
\and
Ehud Lehrer\thanks{\raggedright Economics Department, Durham University, Durham DH1 3LB, UK.  E-mail: \textsf{ehud.m.lehrer@durham.ac.uk}.} \\ {\small Durham University}
\and
Tao Wang\thanks{\raggedright International School of Economics and Management, Capital University of Economics and Business, Beijing 100070, China.  E-mail: \textsf{tao.wang.nau@hotmail.com}.}  \\
{\small CUEB}}
\begingroup
\singlespacing
\maketitle
\endgroup

\thispagestyle{empty}

\begin{abstract}
\singlespacing{
We analyze incomplete-information games where an oracle publicly shares information with players. One oracle dominates another if, in every game, it can match the set of equilibrium outcomes induced by the latter. Characterizations are provided for deterministic signaling functions, based on simultaneous posterior matching, a constructive partition criterion, and common knowledge components. This study elaborates on the work of \citet{Blackwell1951} in games with incomplete information, using the common knowledge components of \citet{Aumann1976}.}
\end{abstract}

\bigskip
\noindent {\emph{Journal of Economic Literature} classification numbers: C72, D82, D83.}


\newpage
\setcounter{page}{1}
\section{Introduction} \label{Section - Intro}

The original motivation for this paper came from the seemingly distant domain of monetary policy, specifically from a public debate surrounding a Federal Open Market Committee press release following the spike in inflation in the aftermath of the COVID pandemic. Commentators and market participants scrutinized every word, tracking deviations from previous statements to extract relevant information. Seeing this debate unfold raised a fundamental economic question: what are the precise strategic implications of a partially informed external agent providing public signals to a market of players who themselves possess heterogeneous, private, and partial information?

This conceptual framework of publicly revealing partial information extends far beyond central banking. It is analogous to the signals provided by various forecasters, ranging from weather and sports to geopolitics, alongside news media organizations, rating agencies, and prediction markets. In each of these settings, an external observer publicly shares incomplete information with a set of strategically interacting agents, thereby reshaping their hierarchies of beliefs and altering the underlying game.

Following \citet{Lagziel2026}, this paper provides a formal framework for analyzing these interactions through incomplete-information games. We assume that players are partially informed about the realized state of the world through their own private partitions. Additional public information is disclosed by an external information provider, which we refer to as an \emph{oracle}. The oracle is endowed with a fixed partition of the state space, representing its informational capability, and communicates via a \textit{signaling function} that is measurable with respect to this partition. Crucially, the oracle may be unaware of the players' private information or what constitutes common knowledge among them. The combination of the original Bayesian game, defined by the players' private information, action sets, and payoff functions, and the additional public information provided by the oracle constitutes what we refer to as a \textit{guided game}.

Unlike a standard player, the oracle acts as a generator of Blackwell experiments. Our primary objective is to establish a preorder of oracles based on their ability to induce equilibria across all possible games. We define \emph{dominance}: one oracle dominates another if for every signaling function of the latter, there exists a signaling function of the former such that for every game $G$ the sets of equilibrium distributions over state-action profiles in the corresponding guided games exactly coincide. The comparison does not impose an equilibrium-selection rule or an objective function for the oracle.

Departing from the original analysis of \cite{Lagziel2026} and \cite{Lagziel2025f}, in this paper we focus strictly on deterministic oracles, where the signaling function maps information to a public signal without randomization. In this  setting, we show that one oracle dominates another if and only if it can replicate the joint posterior beliefs induced by the other oracle across \emph{all players simultaneously}, adjusting for the redundancies created by the players' private information (Theorem~\ref{Theorem - dominance = informativeness}). We formalize this condition as being \emph{Jointly More Informative} (JMI). The JMI condition captures the fact that a single public signal is interpreted differently by different players, therefore the informational contribution of the oracle must be evaluated relative to the heterogeneous prior knowledge of each individual.

While intuitive, the JMI condition represents a fundamental departure from standard partition refinement. The two notions do not coincide: though refinement yields the JMI condition, the reverse derivation does not hold in general.  Nevertheless, we establish a rigid equivalence under two-sided dominance: if the state space consists of a single common knowledge component (the minimal commonly known event; see \citealp{Aumann1976}), and two deterministic oracles dominate each other under the JMI condition, their partitions must be identical (see Theorem~\ref{Theorem: dual JMI implies equivalence}).

We also give a constructive criterion for matching a given signaling function. The criterion identifies the states that Oracle~1 must assign to the same signal, and determines whether these identifications are consistent with every player's target information. To compare the two oracles, it suffices to check the binary signaling functions of Oracle~2.

\subsection{Relation to literature}

This research extends the classical framework established by \citet{Blackwell1951} and \citet{Blackwell1953}, which introduced the standard for comparing experiments in single-agent decision problems. In Blackwell's framework, an experiment dominates another if it yields a weakly higher expected utility for the decision maker across all decision problems. When transitioning to a multi-agent strategic environment, our notion of dominance requires the replication of equilibrium outcome distributions rather than the simple maximization of a single utility function.

Recent advancements have generalized Blackwell's model. Notably, \citet{Brooks2024} compare information sources that are robust to any external information and decision problem. There are critical distinctions between their framework and ours. First, our analysis applies to multi-player games rather than single-agent problems. Second, their characterization is entirely independent of the decision maker's private information. In contrast, our model fixes the players' private information structures while allowing the underlying game payoffs to vary. Consequently, our characterizations are tailored to the specific configuration of the players' prior knowledge. Furthermore, rather than evaluating fixed signals, we treat the oracle as an active generator of experiments constrained by its partition.

This project also parallels the extensive literature on Bayesian persuasion \citep{Kamenica2011, Kamenica2019}. Related work studies information provision, disclosure, and communication in a range of strategic settings \citep{Horner2016, Renault2013, Ganglmair2014, Renault2017, Ely2017, Ely2020, Che2018a, Bizzotto2021, Zhao2024}. Unlike a standard sender, our oracle lacks a specific objective or payoff function in the underlying game. By isolating the informational power of the oracle from any strategic preferences, we provide a pure assessment of informational dominance.

Finally, our framework relates to the study of external mediators in incomplete-information games \citep{Forges1993, Gossner2000}. However, existing literature on mediation frequently focuses on correlating players' actions or generating specific correlated equilibria. Information structures have also been ordered in specific classes of games, such as zero-sum games \citep{Peski2008} or common-interest games \citep{Lehrer2010}. Outcome equivalence in more general games is studied by \citet{Lehrer2013}. \citet{Bergemann2016} compare information structures through individual sufficiency and the sets of Bayes correlated equilibria they induce. These comparisons concern fixed information structures; we instead compare the public experiments available to the oracles. Our approach diverges by restricting the oracle to public signals, precluding private recommendations, while permitting unrestricted variations in the underlying game's payoff structure.
\section{The model} \label{Section - Model}

A \emph{guided game} comprises a Bayesian game and an \emph{oracle}. The oracle's role is to provide information that enables a different, possibly broader, range of equilibria.

We begin by defining the underlying Bayesian game. Let \( N = \{1, 2, \dots, n\} \) be a finite set of \( n \geq 2 \) players, and let \( \Omega \) denote a non-empty, finite state space. Each player \( i \in N \) has a non-empty, finite set of actions \( A_i \) and a partition \( \Pi_i \) over \( \Omega \), representing the information available to player \( i \). Denote the set of action profiles by \( A = \times_{i \in N} A_i \). The utility function for each player \( i \in N \) is \( u_i: \Omega \times A \to \mathbb{R} \), which maps states and action profiles to real-valued payoffs.

To extend the basic game into a guided game, we introduce an oracle who provides public information before players choose their actions. The oracle is endowed with a partition $F$ of the state space $\Omega$, and a countably infinite set $S$ of possible signals.

A \textit{deterministic signaling strategy} is a function $\tau:F\to S$, assigning a single signal to each element of the partition. Note that any deterministic signaling strategy is effectively equivalent to a coarsening of $F$, and we will refer to it as such when appropriate. For any partition $F'$, let $F'(\omega)$ denote its atom containing $\omega$, and write $\tau(\omega)$ for $\tau(F(\omega))$.

The guided game evolves as follows. First, the oracle publicly announces a strategy $\tau$. Then, a state $\omega \in \Omega$ is drawn according to a common prior $\mu \in \Delta (\Omega)$ with full support.\footnote{If the prior does not have full support, we restrict the state space and all information partitions to its support. The posterior-matching requirements then apply on this restricted state space.} Each player $i$ is privately informed of $\Pi_i(\omega)$, which is a set of states containing $\omega$ and also an atom of player $i$'s private partition. Finally, the signal $\tau(\omega) \in S$ is publicly announced, and the players choose their actions. The prior and the information partitions are common knowledge among the players.\footnote{Following \citet{Aumann1976}, a \emph{common knowledge component} (CKC) is an atom of the meet $\bigwedge_{i\in N}\Pi_i$, the finest common coarsening of the players' partitions. Equivalently, it is a minimal event, with respect to set inclusion, that is common knowledge among the players.}

Let the join $\Pi_i \vee F'$ denote the coarsest common refinement of $\Pi_i$ and $F'$, namely the updated information (i.e., partition) of player $i$. Its atoms are the non-empty intersections of their atoms. The meet of partitions is their finest common coarsening. Let $\mu^i_{\tau|\o}  = \mu (\cdot |[\Pi_i \vee \tau](\o) )  \in \Delta(\Omega)$  denote player $i$'s posterior belief after observing $\Pi_i(\omega)$ and $\tau(\omega)$. Every strategy $\tau$ yields an incomplete-information game $G(\tau)$, determined by the prior, the updated partitions $(\Pi_i\vee\tau)_{i\in N}$, the action sets, and the utility functions. Since the state space and the action sets are finite, the equilibria of the game exist. When there is no risk of ambiguity, we denote the incomplete-information game without \( \tau \) by \( G \).

\subsection{A notion of dominance}

To discuss the role of the oracle in the current framework, one needs a relevant solution concept. Thus, let us define the following notion of a Guided equilibrium, which incorporates the oracle's strategy. Formally, let $\sigma_i: \Pi_i\times S\rightarrow \Delta (A_i)$ be a strategy of player $i$. A tuple $(\tau,\sigma_1,\dots,\sigma_n)$ is a \emph{Guided equilibrium} if $(\sigma_1,\dots,\sigma_n)$ is a Nash equilibrium in the incomplete-information game $G(\tau)$.

The notion of a Guided equilibrium defines a preorder of oracles, i.e., a reflexive and transitive relation over their partitions according to the sets of equilibria. To define this relation, let $\rm{D}(G(\tau)) \subseteq \Delta(\Omega \times A)$ be the set of distributions over $\Omega \times A$ induced by Nash equilibria given $G$ and $\tau$. Now consider two oracles, Oracle $1$ and Oracle $2$, and denote the generic partition and strategy of Oracle $j$ by $F_j$ and $\tau_j$, respectively. Throughout the comparison, $N$, $\Omega$, $\mu$, and the players' partitions are fixed, while $G$ varies over all finite action sets and utility functions. Using these notations we define the ordering of oracles as follows.

\begin{definition}[Dominance of deterministic oracles] \label{Definition - Strategic Dominance}
    Fix the players' information structures. We say that \emph{Oracle $1$ deterministically dominates Oracle $2$}, denoted $F_1 \succeq_{\rm{D}} F_2$, if for every deterministic signaling strategy $\tau_2$, there exists a deterministic signaling strategy $\tau_1$ such that for every game $G$, $\rm{D}(G(\tau_1)) = \rm{D}(G(\tau_2))$.
\end{definition}

Two oracles are \emph{equivalent} if each dominates the other. The dominance relation is a partial order on these equivalence classes. This is an analyst's comparison of the available experiments; implementing a particular matching rule may require knowledge of the players' partitions.

\section{Ordering of deterministic oracles} \label{Section - Partial ordering of deterministic oracles}

Our first main result characterizes the notion of dominance among oracles, assuming they are restricted to deterministic strategies.

The characterization is based on the ability of one oracle to \emph{match} the players' joint posterior beliefs, for any given strategy of the other oracle. More formally, we say that Oracle $1$ is \emph{jointly More Informative} than Oracle $2$, if, for every strategy $\tau_2$, there exists a strategy $\tau_1$ that simultaneously matches the posterior partition of every player $i$.

\begin{definition} \label{Definition - jointly More Informative}
        \emph{Oracle $1$ is jointly More Informative  (JMI) than Oracle $2$}, denoted $F_1 \succeq_{(\mu^i)_i} F_2$, if for every deterministic $\tau_2$, there exists a deterministic $\tau_1$ such that $\Pi_i \vee \tau_1 = \Pi_i \vee \tau_2$ for every player $i$.
\end{definition}

\begin{obs} \label{Observation - JMI through the premitives}
    \emph{Oracle $1$} is jointly More Informative (JMI) than \emph{Oracle $2$} if and only if for every coarsening $F_2'$ of $F_2$, there exists a coarsening $F_1'$ of $F_1$ such that $\Pi_i \vee F_1' = \Pi_i \vee F_2'$, for every player $i$.
\end{obs}

Note that Observation \ref{Observation - JMI through the premitives} follows directly from Definition \ref{Definition - jointly More Informative} because every $F_i$-measurable deterministic strategy $\tau_i$ induces a coarsening $F_i'$ of $F_i$ and vice versa.
Nevertheless, what should be clear is that the notion of JMI differs from the notion of refinement, as the following example illustrates.

\begin{example} Jointly More Informative versus refinement. \label{ex: more informative vs refinement} \end{example} 

The ordering generated by the notion of ``jointly More Informative than" need not coincide with the notion of ``finer than". Consider, for example, two players with the same partition $\Pi_1=\Pi_2=\{\{\o_1,\o_2\},\{\o_3,\o_4\}\}$, and oracle partitions $F_1 = \{\{\o_1,\o_2,\o_3\},\{\o_4\}\}$ and $F_2 = \{\{\o_1,\o_2\},\{\o_3\},\{\o_4\}\}$. Note that $F_2$ strictly refines $F_1$ and each player's partition, but Oracle $1$ remains jointly More Informative than Oracle $2$. This is illustrated in Figure \ref{fig:JMI is not finer than}. Every strategy of Oracle~2 either separates $\omega_3$ from $\omega_4$ or leaves them indistinguishable; Oracle~1 matches these cases by fully disclosing its information or withholding it, respectively. The reverse JMI relation follows from refinement. Thus, the two oracles are equivalent despite having different partitions. Nevertheless, in Section \ref{Section - Two-sided JMI implies equivalence} we prove that if $F_1$ is JMI than $F_2$ and vice versa, then the two partitions coincide within every CKC.

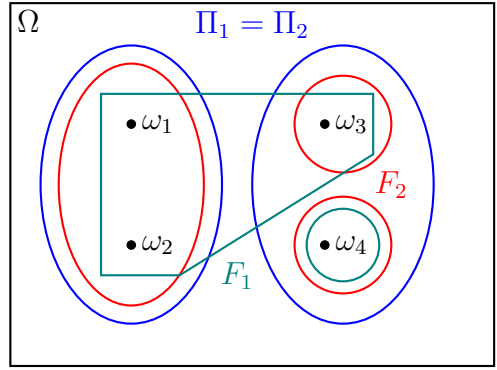
\begin{figure}[htbp]
\centering

{\normalsize Jointly More Informative versus Refinement}

\medskip

\begin{tikzpicture}[scale=0.8] 

\draw[thick] (0,0) rectangle (8,6);

\node at (0.3,5.7) {$\Omega$};

\node[blue] at (4,5.65) {$\Pi_1=\Pi_2$};

\node[red] at (6.3,3) {$F_2$};

\node[teal] at (3.75,1.5) {$F_1$};

\draw[blue, thick] (2,3) ellipse (1.5 and 2.3);

\draw[red, thick] (2,3) ellipse (1.2 and 2);

\draw[red, thick] (5.5,4) ellipse (0.8 and 0.8);
\draw[red, thick] (5.5,2) ellipse (0.8 and 0.8);
\draw[teal, thick] (5.5,2) ellipse (0.6 and 0.6);

\draw[blue, thick] (5.5,3) ellipse (1.5 and 2.3);

\filldraw[black] (2,4) circle (2pt) node[anchor=west] {$\o_1$};
\filldraw[black] (2,2) circle (2pt) node[anchor=west] {$\o_2$};
\filldraw[black] (5.2,4) circle (2pt) node[anchor=west] {$\o_3$};
\filldraw[black] (5.2,2) circle (2pt) node[anchor=west] {$\o_4$};

\draw[teal, thick] (1.5,4.5) -- (6,4.5) -- (6,3.5) -- (2.8,1.5) -- (1.5,1.5) -- cycle;

\end{tikzpicture}
\caption{\footnotesize The notion of ``jointly More Informative than" does not imply ``finer than", though the latter does imply the former. In this figure, $F_2$ (red) strictly refines $F_1$ (green) and $\Pi_1=\Pi_2$ (blue), but for every deterministic $\tau_2$, there exists a deterministic $\tau_1$ such that $\Pi_i\vee \tau_1 = \Pi_i\vee \tau_2$ for every player $i$, so $F_1$ is jointly More Informative than $F_2$.} \label{fig:JMI is not finer than}
\end{figure}

One can also bridge the gap between the notions of JMI and refinement by considering the possibility that the players' partitions are not fixed.\footnote{This resembles the condition of strong Blackwell dominance, in the context of decision problems, in \cite{Brooks2024}.} In other words, we can also consider the possibility that Oracle $1$ is JMI than Oracle $2$ for \emph{any} set of the players' partitions. Once we account for all possible partitions, we must also account for the trivial partition, so that Oracle $1$ must match any deterministic strategy of Oracle $2$. This implies that $F_1$ refines $F_2$, at least weakly. Conversely, refinement allows Oracle~1 to reproduce every signaling strategy of Oracle~2 for any profile of players' partitions.

\begin{example}Simultaneous posterior matching vs posterior matching. \label{ex: simultaneous matching}
\end{example} 

The JMI condition differs from individually matching the players' posteriors. To see this, consider a state space $\Omega=\{\o_1,\o_2,\o_3\}$, and two players with partitions $\Pi_1=\{\{\o_1,\o_2\},\{\o_3\}\}$ and $\Pi_2=\{\{\o_1\},\{\o_2,\o_3\}\}$. The oracles' partitions are $F_1=\{\{\o_1,\o_3\},\{\o_2\}\}$ and $F_2=\{\{\o_1,\o_2\},\{\o_3\}\}$. Now consider full disclosure by Oracle~2. It provides no additional information to player~1, but makes player~2 fully informed. Oracle~1 can match the information of either player separately: it can withhold its information for player~1, or fully disclose its information for player~2. These are its only two informationally distinct strategies, and neither matches both players simultaneously.

\begin{table}[htbp]
\centering
\begin{tabular}{@{}lcc@{}}
\hline
Public information & Player 1 & Player 2\\
\hline
Full disclosure by Oracle 2 & Unchanged & Fully informed\\
No disclosure by Oracle 1 & Unchanged & Unchanged\\
Full disclosure by Oracle 1 & Fully informed & Fully informed\\
\hline
\end{tabular}
\caption{\footnotesize Separate matching does not imply simultaneous matching.}
\label{tab:simultaneous-matching}
\end{table}

Thus, Oracle~1 can match everything Oracle~2 provides to either player considered separately, but it is not JMI than Oracle~2. The difficulty is that the same public signaling strategy must match the information of all players.

\subsection{The characterization} \label{Section - First characterization result - deterministic oracles}

Our main result, given in Theorem \ref{Theorem - dominance = informativeness} below, presents an equivalence between oracle dominance and the notion of being jointly More Informative. Specifically, we prove that one oracle dominates another if and only if it is jointly More Informative.\footnote{The proofs are deferred to the Appendix.}

\begin{theorem} \label{Theorem - dominance = informativeness}
    Assume that oracles are deterministic.
    Then, \emph{Oracle $1$} dominates \emph{Oracle $2$} if and only if \emph{Oracle $1$} is jointly More Informative than \emph{Oracle $2$}.
\end{theorem}

The proof constructs a game from separate decision problems, one for each player. Any failure to match a player's posterior partition changes that player's equilibrium expected payoff. Thus, the game separates the target strategy from every strategy that fails to match all players simultaneously. Note that the same characterization holds if equality in Definition~\ref{Definition - Strategic Dominance} is replaced by $\rm{D}(G(\tau_2))\subseteq\rm{D}(G(\tau_1))$. Indeed, when JMI fails, the proof constructs a game for which the corresponding equilibrium-distribution sets are disjoint.

\subsection{A constructive matching criterion}\label{Section - constructive matching}

The JMI condition requires a matching strategy for each strategy of Oracle~2. We now give a direct way to determine whether such a strategy exists. Fix $\tau_2$, and define
\[ Q(\tau_2)=F_1\wedge\bigwedge_{i\in N}(\Pi_i\vee\tau_2). \]
Recall that the meet is the finest common coarsening. The partition $Q(\tau_2)$ is obtained by identifying states in the same atom of $F_1$ or of any $\Pi_i\vee\tau_2$, and taking the transitive closure of these identifications. Any matching strategy must respect these identifications: it cannot distinguish states that Oracle~1 cannot distinguish, or give a player information beyond that provided by $\tau_2$. The following proposition shows that it suffices to check whether $Q(\tau_2)$ itself is a matching strategy.

\begin{proposition}\label{Proposition - constructive matching}
There exists an $F_1$-measurable deterministic strategy $\tau_1$ such that $\Pi_i\vee\tau_1=\Pi_i\vee\tau_2$ for every player $i$ if and only if $\Pi_i\vee Q(\tau_2)=\Pi_i\vee\tau_2$ for every player $i$. Whenever this condition holds, $Q(\tau_2)$ itself is a matching strategy.
\end{proposition}

Oracle~1 must send the same signal at states in the same atom of $F_1$. It must also send the same signal at states that a player cannot distinguish under $\tau_2$, otherwise it would provide that player with additional information. If these requirements, taken together, identify two states in the same atom of some $\Pi_i$ that $\tau_2$ separates, matching is impossible. Otherwise, Oracle~1 can assign a different signal to each atom of $Q(\tau_2)$.

In Example~\ref{ex: simultaneous matching}, Oracle~1's partition identifies $\omega_1$ with $\omega_3$, while preserving player~1's information under full disclosure by Oracle~2 requires identifying $\omega_1$ with $\omega_2$. Thus, $Q(\tau_2)=\{\Omega\}$, which cannot make player~2 fully informed.

Any deterministic strategy can be obtained by combining binary disclosures, each separating one of its atoms from its complement. The following corollary shows that matching these disclosures is enough, since their matching strategies can also be combined into a single $F_1$-measurable strategy.

\begin{corollary}\label{Corollary - binary tests}
Oracle~1 is JMI than Oracle~2 if and only if it can simultaneously match every player's posterior partition under every strategy of Oracle~2 that uses at most two signals. The matching strategies of Oracle~1 need not be binary.
\end{corollary}

This reduces the comparison to binary strategies of Oracle~2, with each one checked by Proposition~\ref{Proposition - constructive matching}.

These results are related to the information loops studied in \citet{Lagziel2026}. In both settings, the oracle's measurability restriction may prevent locally feasible information from being provided simultaneously. In the deterministic case, the restriction takes the form of forced identifications: a chain of states that must receive the same signal may connect two states that a player's target information requires to distinguish. Proposition~\ref{Proposition - constructive matching} identifies all such chains, and Corollary~\ref{Corollary - binary tests} shows that binary disclosures suffice to detect any obstruction. Unlike the stochastic loop conditions, however, these requirements concern distinctions between states rather than their relative likelihoods.

\subsection{JMI and refinement within CKCs} \label{Section - JMI and refinement within CKCs}

Example~\ref{ex: more informative vs refinement} shows that JMI does not imply partition refinement. This raises the question of whether JMI instead corresponds to refinement within every CKC, namely refinement after restricting the partitions to each component. The following two examples show that neither implication holds.

\begin{example}[JMI does not imply refinement in every CKC]\label{Example: JMI does not imply refinement} \end{example} 
To see that JMI does not imply refinement in every CKC, consider the information structure described in Figure~\ref{fig:unique-CKC-counterexample}. There are four states and three players, with partitions
\[
\begin{aligned}
\Pi_1&=\{\{\omega_1,\omega_4\},\{\omega_2\},\{\omega_3\}\},\\
\Pi_2&=\{\{\omega_1\},\{\omega_2\},\{\omega_3,\omega_4\}\},\\
\Pi_3&=\{\{\omega_1\},\{\omega_2,\omega_3\},\{\omega_4\}\}.
\end{aligned}
\]
There is a unique CKC. The oracles' partitions are $F_1=\{\{\omega_1,\omega_2\},\{\omega_3\},\{\omega_4\}\}$ and $F_2=\{\{\omega_1,\omega_3\},\{\omega_2,\omega_4\}\}$.

Both oracles can either withhold all information or fully disclose their information, thereby ensuring that all players become fully informed of the realized state in the latter case. These are all the possible experiments of Oracle~2. On the other hand, Oracle~1 can also signal the partition $F_1'=\{\{\omega_1,\omega_2,\omega_3\},\{\omega_4\}\}$, which provides complete information to players~1 and~2 but provides no information to player~3. Thus, Oracle~1 is JMI than Oracle~2, while neither of the two partitions is finer than the other. The converse JMI relation does not hold, because Oracle~2 cannot replicate $F_1'$.

\begin{figure}[htbp]
\centering
\begin{minipage}[t]{.32\textwidth}
\centering
\begin{tikzpicture}[scale=.70]

\draw[thick] (0,0) rectangle (6,6);

\node at (0.3,5.7) {$\Omega$};

\node[blue] at (1.8,5.6) {$\Pi_1$};
\node[red] at (3.8,5.6) {$\Pi_2$};
\draw[red, thick] (1.5,4.5) ellipse (0.8 and 0.8);
\draw[blue, thick] (1.5,1.5) ellipse (0.6 and 0.6);
\draw[red, thick] (1.5,1.5) ellipse (0.8 and 0.8);
\draw[red, thick] (4.5,3) ellipse (1 and 2.5);
\draw[blue, thick] (4.5,4.5) ellipse (0.6 and 0.6);
\draw[blue, thick, rotate around={135:(3,3)}] (3,3) ellipse (3.2 and 1.2);

\filldraw[black] (1.4,4.5) circle (2pt) node[anchor=west] {$\o_1$};
\filldraw[black] (1.4,1.5) circle (2pt) node[anchor=west] {$\o_2$};
\filldraw[black] (4.4,4.5) circle (2pt) node[anchor=west] {$\o_3$};
\filldraw[black] (4.4,1.5) circle (2pt) node[anchor=west] {$\o_4$};

\end{tikzpicture}
\par (a) $\Pi_1$ and $\Pi_2$
\end{minipage}
\hfill
\begin{minipage}[t]{.32\textwidth}
\centering
\begin{tikzpicture}[scale=.70]

\draw[thick] (0,0) rectangle (6,6);

\node at (0.3,5.7) {$\Omega$};

\node[black] at (5.3,5.6) {$\Pi_3$};

\draw[black, thick] (1.5,4.5) ellipse (0.6 and 0.6);
\draw[black, thick] (4.5,1.5) ellipse (0.6 and 0.6);
\draw[black, thick, rotate around={45:(3,3)}] (3,3) ellipse (3.2 and 1.2);

\filldraw[black] (1.4,4.5) circle (2pt) node[anchor=west] {$\o_1$};
\filldraw[black] (1.4,1.5) circle (2pt) node[anchor=west] {$\o_2$};
\filldraw[black] (4.4,4.5) circle (2pt) node[anchor=west] {$\o_3$};
\filldraw[black] (4.4,1.5) circle (2pt) node[anchor=west] {$\o_4$};

\end{tikzpicture}
\par (b) $\Pi_3$
\end{minipage}
\hfill
\begin{minipage}[t]{.32\textwidth}
\centering
\begin{tikzpicture}[scale=0.70]

\draw[thick] (0,0) rectangle (6,6);
\node at (0.3,5.7) {$\Omega$};
\node[teal] at (3,5.2) {$F_2$};
\node[orange] at (2.8,3) {$F_1$};
\draw[teal, thick] (3,4.5) ellipse (2.5 and 1);
\draw[teal, thick] (3,1.5) ellipse (2.5 and 1);
\draw[orange, thick] (4.5,4.5) ellipse (0.6 and 0.6);
\draw[orange, thick] (4.5,1.5) ellipse (0.6 and 0.6);
\draw[orange, thick] (1.5,3) ellipse (1 and 2.5);

\filldraw[black] (1.5,4.5) circle (2pt) node[anchor=west] {$\o_1$};
\filldraw[black] (1.5,1.5) circle (2pt) node[anchor=west] {$\o_2$};
\filldraw[black] (4.3,4.5) circle (2pt) node[anchor=west] {$\o_3$};
\filldraw[black] (4.3,1.5) circle (2pt) node[anchor=west] {$\o_4$};

\end{tikzpicture}
\par (c) $F_1$ and $F_2$
\end{minipage}
\caption{ \footnotesize JMI does not imply refinement in a unique CKC. Panels (a) and (b) show the players' information, and panel (c) shows the oracles' information. Oracle~1 is JMI than Oracle~2, while neither oracle's partition refines the other's.}
\label{fig:unique-CKC-counterexample}
\end{figure}
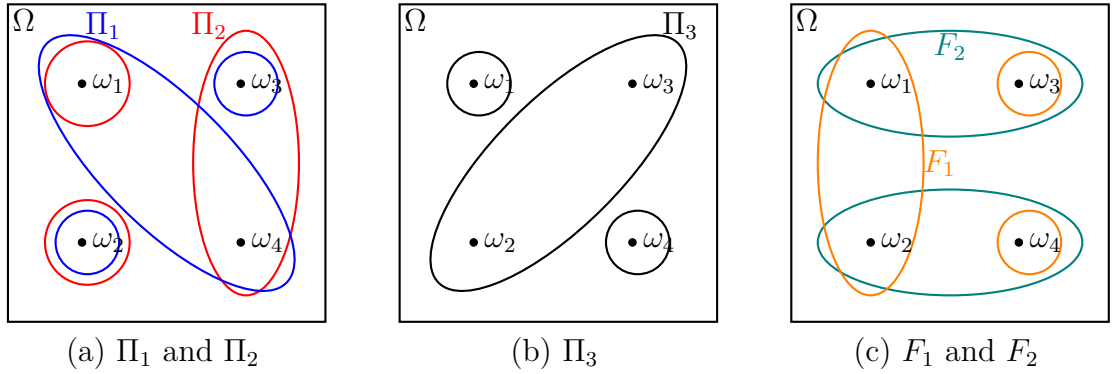

\begin{example}[Refinement in every CKC does not imply JMI]\label{Example: refinement does not imply JMI}
\end{example} 
To demonstrate that refinement in every CKC does not imply JMI, consider two players whose partitions are
\[
\begin{aligned}
\Pi_1&=\{\{\omega_1,\omega_2\},\{\omega_4,\omega_5\},\{\omega_3,\omega_6\}\},\\
\Pi_2&=\{\{\omega_1,\omega_2\},\{\omega_3,\omega_4\},\{\omega_5,\omega_6\}\}.
\end{aligned}
\]
In this case, there are two CKCs, $C_1=\{\omega_1,\omega_2\}$ and $C_2=\{\omega_3,\omega_4,\omega_5,\omega_6\}$. The oracles' partitions are $F_1=\{\{\omega_1,\omega_3,\omega_4\},\{\omega_2,\omega_5,\omega_6\}\}$ and $F_2=\{\{\omega_1,\omega_2\},\{\omega_3,\omega_4\},\{\omega_5,\omega_6\}\}$, as illustrated in Figure~\ref{fig:two-CKC-counterexample}. Observe that in every CKC, $F_1$ refines $F_2$.

Now consider a completely revealing, deterministic experiment $\tau_2$ that maps the three different partition elements of $F_2$ to three different signals. Can Oracle~1 produce an experiment $\tau_1$ such that $\Pi_i\vee\tau_1=\Pi_i\vee\tau_2$ for every player $i$?

Note that under $\tau_2$, neither player can distinguish $\omega_1$ from $\omega_2$. Therefore, in order for $\tau_1$ to satisfy $\Pi_i\vee\tau_1=\Pi_i\vee\tau_2$ for every $i$, the experiment $\tau_1$ must map all $F_1$ partition elements to the same signal. Consequently, under $\tau_1$, player~1 cannot distinguish $\omega_4$ from $\omega_5$, which is achievable given $\tau_2$. We therefore conclude that Oracle~1 is not JMI than Oracle~2, even though $F_1$ refines $F_2$ in every CKC.

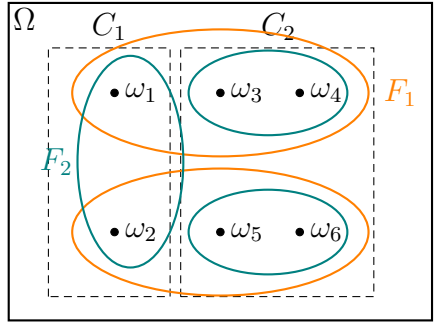
\begin{figure}[htbp]
\centering
\begin{tikzpicture}[scale=0.70] 

\draw[thick, black] (0,0) rectangle (8,6);

\node at (0.3,5.7) {$\Omega$};
\draw[black, densely dashed] (0.75,0.45) rectangle (3.05,5.15);
\draw[black, densely dashed] (3.25,0.45) rectangle (6.9,5.15);
\node at (1.9,5.55) {$C_1$};
\node at (5.1,5.55) {$C_2$};

\draw[orange, thick] (4,4.3) ellipse (2.8 and 1.2);
\draw[orange, thick] (4,1.67) ellipse (2.8 and 1.2);
\node[orange] at (7.4,4.3) {$F_1$};

\draw[teal, thick] (2.3,3) ellipse (1 and 2);
\draw[teal, thick] (4.9,4.3) ellipse (1.5 and 0.8);
\draw[teal, thick] (4.9,1.67) ellipse (1.5 and 0.8);
\node[teal] at (0.9,3) {$F_2$};

\filldraw[black] (2,4.3) circle (2pt) node[anchor=west] {$\o_1$};
\filldraw[black] (2,1.67) circle (2pt) node[anchor=west] {$\o_2$};
\filldraw[black] (4,4.3) circle (2pt) node[anchor=west] {$\o_3$};
\filldraw[black] (5.5,4.3) circle (2pt) node[anchor=west] {$\o_4$};
\filldraw[black] (4,1.67) circle (2pt) node[anchor=west] {$\o_5$};
\filldraw[black] (5.5,1.67) circle (2pt) node[anchor=west] {$\o_6$};

\end{tikzpicture}
\caption{ \footnotesize Refinement in every CKC does not imply JMI. The dashed rectangles mark the two CKCs. Despite $F_1$ (orange) refining $F_2$ (teal) in each component, Oracle~1 is not JMI than Oracle~2.}
\label{fig:two-CKC-counterexample}
\end{figure}

\section{Two-sided JMI implies equivalence in every CKC} \label{Section - Two-sided JMI implies equivalence}

The preceding examples show that refinement within every CKC is neither necessary nor sufficient for JMI. Nevertheless, mutual JMI is more restrictive within a CKC. For a partition $F$ and a CKC $C$, write $F|_C=\{B\cap C:B\in F,\ B\cap C\neq\varnothing\}$. Restricting the model to $C$ means using these restricted partitions, the conditional prior $\mu(\cdot\mid C)$, and all coarsenings of the restricted oracle partitions.

The following theorem provides this equivalence by stating that, given a specific CKC, both restricted oracles dominate each other if and only if their partitions coincide.

\begin{theorem} \label{Theorem: dual JMI implies equivalence}
Fix a \emph{CKC} $C$ and restrict the model to it. Within this \emph{CKC}, each oracle is \emph{JMI} than the other if and only if $F_1|_C=F_2|_C$.
\end{theorem}

In other words, the theorem asserts that the restricted partitions coincide if and only if the two restricted oracles are mutually JMI. In particular, when $\Omega$ is a single CKC, mutual dominance is equivalent to $F_1=F_2$. The componentwise statement compares the restricted models and does not assert that local equivalence is sufficient for global dominance.

\bibliographystyle{chicago}
\begingroup
\interlinepenalty=10000
\bibliography{references2}
\endgroup

\appendix

\section{Proofs}

\subsection{Proof of Theorem \ref{Theorem - dominance = informativeness}}

\begin{proof}
One derivation is straightforward.
Assume that $F_1 \succeq_{(\mu^i)_i} F_2$.
For every $\tau_2$, take $\tau_1$ such that $\Pi_i\vee \tau_1= \Pi_i\vee \tau_2$ for every player $i$.
Thus, we get $ \rm{D}(G(\tau_1))= \rm{D}(G(\tau_2))$ for every game $G$.
This holds for every strategy $\tau_2$, so $F_1 \succeq_{\rm{D}} F_2 $ as needed.

To establish the converse derivation of the theorem, we assume that Oracle $1$ is not jointly More Informative than Oracle $2$, and prove that Oracle $1$ does not  dominate Oracle $2$. Fix a strategy $\tau_2$, so that for every $\tau_1$, there exists a player $i$ such that $\Pi_i\vee \tau_1 \neq \Pi_i\vee \tau_2$. For every player $i$, we construct a decision problem using the fixed prior $\mu$ and $\Pi_i\vee\tau_2$, and combine these problems into a game in which each player's payoff depends only on the state and his own action. We give the construction and computation for player~1. Denote $\Pi_1\vee \tau_2 =\{B_1,\dots,B_k\} $ where $B_j=\{\omega^j_1,\dots,\omega^j_{|B_j|}\} \subseteq \Omega$ for every $1 \leq j \leq k$.

Consider the following decision problem. Define $P_{B_j}$ to be the set of all permutations of $B_j$, so that every element $p\in P_{B_j}$ is a function $p:B_j \to \{1,2,\dots, |B_j|\}$ where $p(\omega^j_l)$ is the location of $\omega^j_l$ according to that permutation. Let $A_1=\bigcup_{j}P_{B_j}$ be the action set of player $1$, so that player $1$ chooses a permutation $p$ over a partial set of $\Omega$. Define the following utility function, where $p=a_1$:
\begin{equation*}
    u_1(\omega^j_l,a) = u_1(p,\omega^j_l) =
    \begin{cases}
        \frac{p(\omega^j_l)}{\mu(\omega^j_l|B_j)|B_j|}, & \text{if } p \in P_{B_j}, \\
        -\frac{2^{10|\Omega|}}{\min_{\omega} \mu(\omega)}, & \text{if } p \notin P_{B_j},
    \end{cases}
\end{equation*}
where $\mu(\omega^j_l|B_j)$ is the probability of $\omega^j_l$ conditional on $B_j$, and $u_1(p,\omega)$ is shorthand for the payoff when player~1 chooses $p$. In simple terms, player $1$ needs to match his action, i.e., a permutation, to the realized state $\omega^j_l$. If the action of player $1$ is not a permutation on the states of the realized element of the partition (generated by his private information and the information that Oracle $2$ conveys), he gets an extremely low negative payoff. However, in case the action of player $1$ is a permutation on the relevant block, he receives a positive payoff based on the ordinal location of the realized state according to the chosen permutation.

Let us compare the expected payoffs of player $1$ given the additional information conveyed separately by the two oracles. Given the partition $\Pi_1\vee \tau_2$ and after $\omega$ is realized, player $1$ is informed of the relevant block $B_j$ of the partition such that $\omega \in B_j$. Thus, for every $p\in P_{B_j}$,
\[
\begin{aligned}
\mathbf{E}[u_1(p,\omega)|B_j]
&=\sum_{\omega^j_l\in B_j}\mu(\omega^j_l|B_j)u_1(p,\omega_l^j)\\
&=\sum_{\omega^j_l\in B_j}\mu(\omega^j_l|B_j)
  \frac{p(\omega^j_l)}{\mu(\omega^j_l|B_j)|B_j|}
 =\sum_{\omega^j_l\in B_j}\frac{p(\omega^j_l)}{|B_j|}
 =\frac{|B_j|+1}{2}.
\end{aligned}
\]
Note that the expected payoff is independent of the chosen permutation $p$ given that $p\in P_{B_j}$. Here and below, maximization given a partition allows the chosen action to depend on its realized atom, and the value is evaluated ex ante. Hence,
\[
\max_{p}\mathbf{E}[u_1(p,\omega)|\Pi_1\vee \tau_2] = \sum_{j=1}^k \mu(B_j) \frac{|B_j|+1}{2}.
\]

Now fix any $\tau_1$ for which $\Pi_1 \vee \tau_1 \neq \Pi_1\vee \tau_2$. There are two possible scenarios: either $\Pi_1\vee \tau_1$ is a strict refinement of $\Pi_1\vee \tau_2$, or there exists at least one block of $\Pi_1\vee \tau_1$ that intersects two disjoint blocks of $\Pi_1\vee \tau_2$.

Starting with the former, assume that $\Pi_1\vee \tau_1$ is a strict refinement of $\Pi_1\vee \tau_2$, so there exists a block $B_j^*$ that $\Pi_1\vee \tau_1$ splits into at least two separate blocks. Without loss of generality, assume that $B_1$ is such a block, and denote two non-empty unions of the resulting cells by $B_{1,1}$ and $B_{1,2}$, chosen so that $B_1=B_{1,1}\cup B_{1,2}$ and $B_{1,1}\cap B_{1,2}=\varnothing$. The player can condition his action on these two sets. Assume that for every $B_{j} \neq B_1$, player $1$ follows the same strategy as with $\Pi_1\vee \tau_2$ so that we can focus on the difference in expected payoffs given $B_1$. Evidently,
\begin{align*}
    \mathbf{E}[u_1(p,\omega)|B_{1,1}]
    &=  \sum_{\omega^1_l\in B_{1,1}} \mu(\omega^1_l|B_{1,1}) u_1(p,\omega_l^1)
      = \sum_{\omega^1_l\in B_{1,1}} \mu(\omega^1_l|B_{1,1}) \frac{p(\omega^1_l)}{\mu(\omega^1_l|B_1)|B_1|} \\
    &=  \sum_{\omega^1_l\in B_{1,1}} \mu(\omega^1_l|B_{1}) \frac{\mu(B_{1})}{\mu(B_{1,1})}\cdot  \frac{p(\omega^1_l)}{\mu(\omega^1_l|B_1)|B_1|} \\
    &=  \frac{\mu(B_{1})}{\mu(B_{1,1})|B_1|} \sum_{\omega^1_l\in B_{1,1}} p(\omega^1_l).
\end{align*}
Note that player $1$ can choose a permutation on $B_1$ which maximizes the sum of the ranks of all states in $B_{1,1}$. Assigning them the highest ranks gives $\sum_{\omega\in B_{1,1}}p(\omega)>|B_{1,1}|(|B_1|+1)/2$, since $B_{1,1}$ is a proper subset of $B_1$. Thus,
\[
\max_{p\in P_{B_1}}\mathbf{E}[u_1(p,\omega)|B_{1,1}] > \frac{\mu(B_{1})}{\mu(B_{1,1})|B_1|} |B_{1,1}| \frac{|B_1|+1}{2},
\]
and a similar computation holds for $B_{1,2}$. Therefore,
\begin{align*}
    \max_p \mathbf{E}[u_1(p,\omega)|\Pi_1 \vee \tau_1]
    &>  \sum_{j=1}^k \mu(B_j) \frac{|B_j|+1}{2} = \max_{p}\mathbf{E}[u_1(p,\omega)|\Pi_1\vee \tau_2],
\end{align*}
and player $1$ can guarantee a strictly higher expected payoff using the information conveyed through Oracle $1$ than through Oracle $2$.

Next, consider the other possibility that $\Pi_1\vee \tau_1$ is not a refinement of $\Pi_1\vee \tau_2$. This implies that there exists at least one block of $\Pi_1\vee \tau_1$ that intersects two disjoint blocks of $\Pi_1\vee \tau_2$. Denote this block by $B^*$. For any permutation chosen on $B^*$, at least one state in $B^*$ belongs to a different target block and yields the negative payoff. Its contribution to the unconditional expected payoff is at most $-2^{10|\Omega|}$. At every state with a positive payoff, $\mu(\omega)u_1(p,\omega)\leq1$. Thus, for every action $p$,
\[
\mu(B^*)\mathbf{E}[u_1(p,\omega)|B^*]
\leq |B^*|-2^{10|\Omega|}.
\]
The same bound holds for mixed actions. The contribution from states outside $B^*$ is at most $|\Omega|-|B^*|$. This suggests that the expected payoff of player $1$ given $\Pi_1\vee\tau_1$ is bounded from above by
\[
\max_p\mathbf{E}[u_1(p,\omega)|\Pi_1\vee\tau_1]
\leq |\Omega|-2^{10|\Omega|}<0,
\]
which is strictly below the expected payoff given the information transmitted through Oracle $2$.

To conclude, we have defined, for every player $i$, a decision problem independently of $\tau_1$ such that whenever $\Pi_i \vee \tau_1 \neq \Pi_i \vee \tau_2$, it follows that the expected payoff of player $i$ given $\tau_2$ differs from the player's expected payoff given $\tau_1$. Since the players' decision problems are independent, every strategy of either oracle yields a unique profile of equilibrium expected payoffs. For every $\tau_1$, at least one player's partition differs from that under $\tau_2$, so this payoff profile cannot be matched. Thus, $\rm{D}(G(\tau_2))\cap\rm{D}(G(\tau_1))=\varnothing$ for every $\tau_1$, and this concludes the proof.
    \hfill
\end{proof}

\subsection{Proof of Proposition~\ref{Proposition - constructive matching}}
\begin{proof}
Suppose a matching strategy $\tau_1$ exists. It is a coarsening of $F_1$ and of every $\Pi_i\vee\tau_2$, since $\Pi_i\vee\tau_1=\Pi_i\vee\tau_2$. Therefore, $\tau_1$ is a coarsening of $Q(\tau_2)$, and for every player $i$, $\Pi_i\vee Q(\tau_2)$ refines $\Pi_i\vee\tau_1 = \Pi_i\vee\tau_2$. At the same time, $Q(\tau_2)$ is a coarsening of $\Pi_i\vee\tau_2$, so $\Pi_i\vee\tau_2$ refines $\Pi_i\vee Q(\tau_2)$. Thus the two partitions coincide. Conversely, $Q(\tau_2)$ is a coarsening of $F_1$ and therefore an admissible strategy, so the stated equalities give a matching strategy directly.
\end{proof}

\subsection{Proof of Corollary~\ref{Corollary - binary tests}}
\begin{proof}
Necessity is immediate. For sufficiency, fix a strategy $\tau_2$ inducing a partition $\{E_1,\ldots,E_k\}$. If $k=1$, Oracle~1 can send a constant signal. Otherwise, for each $j$, consider the binary strategy $\rho_j=\{E_j,E_j^c\}$, which is a coarsening of $F_2$. By assumption, there exists a coarsening $\tau_1^j$ of $F_1$ such that $\Pi_i\vee\tau_1^j=\Pi_i\vee\rho_j$ for every $i$. The join $\tau_1=\bigvee_{j=1}^k\tau_1^j$ is also a coarsening of $F_1$, and
\[
\Pi_i\vee\tau_1
=\bigvee_{j=1}^k(\Pi_i\vee\tau_1^j)
=\bigvee_{j=1}^k(\Pi_i\vee\rho_j)
=\Pi_i\vee\tau_2
\]
for every player $i$, as needed.
\end{proof}

\subsection{Proof of Theorem \ref{Theorem: dual JMI implies equivalence}}
\begin{proof}
Fix a CKC and restrict the model to it. In this proof, $\Omega$ denotes the fixed component and $F_1,F_2$ denote the restricted partitions, so complements are taken within this component. One direction is trivial, so assume that $F_i$ is JMI than $F_{-i}$ for every $i=1,2$, and let us prove that $F_1=F_2$. Assume, to the contrary, that $F_1\neq F_2$. W.l.o.g., there exist $\o_1\neq \o_2$ such that $F_1(\o_1) = F_1(\o_2)$ whereas $F_2(\o_1) \neq F_2(\o_2)$. Consider the partition $F_2'=\{F_2(\o_1), (F_2(\o_1))^c\}$. By assumption, there exists a coarsening \(F_1'\) of \(F_1\) such that $\Pi_i \vee F_1' = \Pi_i \vee F_2'$, for every player $i$. Denote $A= F_1'(\o_1) \cap F_2(\o_1)$, $B= F_1'(\o_1) \cap (F_2(\o_1))^c$, $C= (F_1'(\o_1))^c \cap (F_2(\o_1))^c$, $D= (F_1'(\o_1))^c \cap F_2(\o_1)$, \(E=F_2(\o_1)\), and \(P_0=F_1'(\o_1)\). Since \(F_1'\) is a coarsening of \(F_1\), we have \(\o_1,\o_2\in P_0\), \(\o_1\in A\), and \(\o_2\in B\). See Figure~\ref{fig:deterministic-equivalence-proof}.

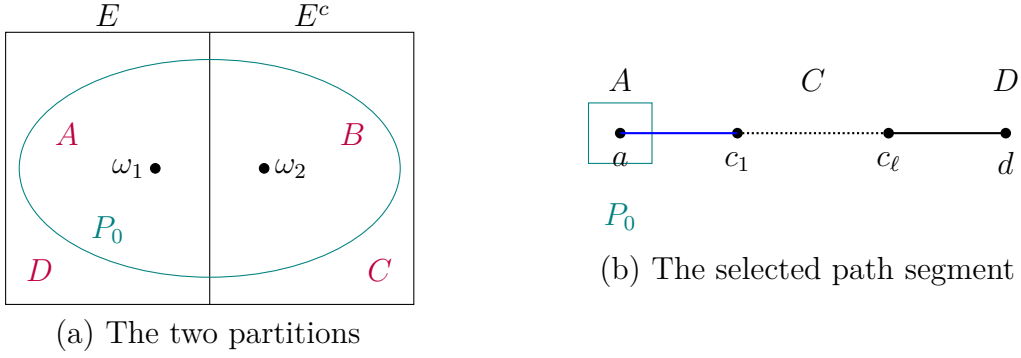
\begin{figure}[ht!]
\centering
\begin{minipage}[c]{.47\textwidth}
\centering
\begin{tikzpicture}[scale=.90]
        \draw[black] (0,0) rectangle (6,4);

        \draw[black] (3,0) -- (3,4);

        \node at (1.5,4.25) {$E$};
        \node at (4.5,4.25) {$E^c$};

        \node[purple] at (0.9,2.5) {\(A\)};
        \node[purple] at (5.1,2.5) {\(B\)};

        \draw[teal] (3,2) ellipse (2.8 and 1.6);
        \node[teal] at (1.5,1.1) {\(P_0\)};

        \filldraw[black] (2.2,2) circle (2pt) node[anchor=east] {\(\omega_1\)};
        \filldraw[black] (3.8,2) circle (2pt) node[anchor=west] {\(\omega_2\)};

        \node[purple] at (0.5,0.5) {\(D\)};
        \node[purple] at (5.5,0.5) {\(C\)};

    \end{tikzpicture}
\par (a) The two partitions
\end{minipage}\hfill
\begin{minipage}[c]{.49\textwidth}
\centering
\begin{tikzpicture}[x=1cm,y=1cm]
\node at (0,1.20) {$A$};
\node at (2.55,1.20) {$C$};
\node at (5.1,1.20) {$D$};
\draw[teal] (-.42,.1) rectangle (.42,.9);
\node[teal] at (0,-.60) {$P_0$};
\fill (0,.5) circle (2pt) node[below=3pt] {$a$};
\fill (1.55,.5) circle (2pt) node[below=3pt] {$c_1$};
\fill (3.55,.5) circle (2pt) node[below=3pt] {$c_\ell$};
\fill (5.1,.5) circle (2pt) node[below=3pt] {$d$};
\draw[blue,thick] (0,.5)--(1.55,.5);
\draw[black,densely dotted,thick] (1.55,.5)--(3.55,.5);
\draw[black,thick] (3.55,.5)--(5.1,.5);
\end{tikzpicture}
\par (b) The selected path segment
\end{minipage}
\caption{ \footnotesize The partition construction and path in the proof of Theorem~\ref{Theorem: dual JMI implies equivalence}. Here $E=F_2(\omega_1)$ and $P_0=F_1'(\omega_1)$. In (b), all intermediate states lie in $C$; only the first edge crosses $P_0$.} \label{fig:deterministic-equivalence-proof}
\end{figure}

Call two states adjacent if they belong to the same partition atom of some player.  For any two adjacent states \(\o,\o'\), equality \(\Pi_i\vee F_1'=\Pi_i\vee F_2'\) for a player \(i\) whose information cell contains both states implies that \(F_1'\) separates \(\o\) and \(\o'\) if and only if \(F_2'\) does. Consequently, there are no player-information edges between \(A\) and \(B\), between \(A\) and \(D\), or between \(B\) and \(C\).

Since the CKC is unique, choose a player-information path from \(\o_1\in A\) to \(\o_2\in B\). Before its first visit to \(B\), this path must visit \(D\): the state immediately preceding the first \(B\)-state can belong neither to \(A\) nor to \(C\). Let \(d\) be the first state of the path in \(D\), and let \(a\) be the last preceding state in \(A\). All intervening states belong to \(C\), so the path contains a segment $(a,c_1,\ldots,c_\ell,d)$, where \(a\in A\), \(c_t\in C\) for every \(t\), and \(d\in D\). Moreover, \(\ell\geq1\), since there is no edge between \(A\) and \(D\).

Now consider the binary partition $F_1''=\{P_0,P_0^c\}$. Because \(F_1'\) is a coarsening of \(F_1\), \(F_1''\) is also a coarsening of \(F_1\). By the reverse JMI relation, there exists a coarsening \(F_2''\) of \(F_2\) such that \(\Pi_i\vee F_1''=\Pi_i\vee F_2''\) for every player \(i\). Since \(a,d\in E=F_2(\o_1)\) and \(F_2''\) is a coarsening of \(F_2\), we have \(F_2''(a)=F_2''(d)\). On the other hand, \(F_1''\) separates \(a\) from \(c_1\), but does not separate any subsequent adjacent pair in the segment \(c_1,\ldots,c_\ell,d\). Equality of the joined partitions along each adjacent pair therefore implies \( F_2''(a)\neq F_2''(c_1) = \cdots=F_2''(c_\ell)=F_2''(d)\), a contradiction. Hence \(F_1\) refines \(F_2\). Interchanging the roles of the two oracles gives the reverse refinement, and therefore \(F_1=F_2\).
\end{proof}

\end{document}